\documentclass[fleqn]{llncs}

\usepackage{latexsym}
\usepackage{amssymb}
\usepackage{stmaryrd}
\usepackage{graphicx}
\usepackage{color}
\usepackage{url}
\usepackage{phonetic}
\usepackage{amsmath}

\newcommand{\be}{\begin{enumerate}}
\newcommand{\ee}{\end{enumerate}}
\newcommand{\bi}{\begin{itemize}}
\newcommand{\ei}{\end{itemize}}
\newcommand{\bc}{\begin{center}}
\newcommand{\ec}{\end{center}}
\newcommand{\bsp}{\begin{sloppypar}}
\newcommand{\esp}{\end{sloppypar}}

\newtheorem{thm}{Theorem}[subsection]

\newtheorem{prop}[thm]{Proposition}
\newtheorem{rem}[thm]{Remark}

\newcommand{\sC}{\mbox{$\cal C$}}
\newcommand{\sD}{\mbox{$\cal D$}}
\newcommand{\sE}{\mbox{$\cal E$}}

\newcommand{\sM}{\mbox{$\cal M$}}

\newcommand{\sT}{\mbox{$\cal T$}}
\newcommand{\sV}{\mbox{$\cal V$}}

\renewcommand{\phi}{\varphi}

\newcommand{\churchqe}{$\mbox{\sc ctt}_{\rm qe}$}
\newcommand{\qzero}{${\cal Q}_0$}

\newcommand{\qzerouqe}{${\cal Q}^{\rm uqe}_{0}$}

\newcommand{\set}[1]{{\{ #1 \}}}
\newcommand{\sembrack}[1]{\llbracket#1\rrbracket}
\newcommand{\synbrack}[1]{\ulcorner#1\urcorner}
\newcommand{\commabrack}[1]{\lfloor#1\rfloor}
\newcommand{\mname}[1]{\mbox{\sf #1}}

\newcommand{\mdot}{\mathrel.}
\newcommand{\tarrow}{\rightarrow}
\newcommand{\LambdaApp}{\lambda\,}
\newcommand{\Neg}{\neg}

\newcommand{\Andd}{\wedge}
\newcommand{\Implies}{\supset}
\newcommand{\Or}{\vee}

\newcommand{\ForallApp}{\forall\,}

\newcommand{\ForsomeApp}{\exists\,}

\newcommand{\TRUE}{\mbox{{\sc t}}}
\newcommand{\FALSE}{\mbox{{\sc f}}}

\title{Incorporating Quotation and Evaluation into Church's Type
  Theory:\\ Syntax and Semantics\thanks{Published in: M. Kohlhase et
    al., eds, \emph{Intelligent Computer Mathematics (CICM 2016)},
    \emph{Lecture Notes in Computer Science}, Vol.~9791, pp.~83--98,
    Springer, 2016. The final publication is available at Springer via
    http://dx.doi.org/10.1007/978-3-319-42547-4\_7. This research was
    supported by NSERC.}}

\author{William M. Farmer}

\institute{%
Computing and Software, McMaster University, Canada\\
\email{wmfarmer@mcmaster.ca}\\[1.5ex]
27 July 2016
}

\begin{document}

\maketitle

\begin{abstract}
{\churchqe} is a version of Church's type theory that includes
quotation and evaluation operators that are similar to quote and eval
in the Lisp programming language.  With quotation and evaluation it is
possible to reason in {\churchqe} about the interplay of the syntax
and semantics of expressions and, as a result, to formalize
syntax-based mathematical algorithms.  We present the syntax and
semantics of {\churchqe} and give several examples that illustrate the
usefulness of having quotation and evaluation in {\churchqe}.  We do
not give a proof system for {\churchqe}, but we do sketch what a proof
system could look like.
\end{abstract}

\iffalse
\textbf{Keywords:} Church's type theory, metareasoning, reflection,
quotation, evaluation, quasiquotation, schemas, syntax-based
mathematical algorithms, meaning formulas, substitution.
\fi
    
\section{Introduction}

The Lisp programming language is famous for its use of
\emph{quotation} and \emph{evaluation}.  From code the Lisp quotation
operator called \emph{quote} produces meta-level data (i.e.,
S-expressions) that represents the code, and from this data the Lisp
evaluation operator called \emph{eval} produces the code that the data
represents.  In Lisp, \emph{metaprogramming} (i.e., programming at the
meta-level) is performed by manipulating S-expressions and is
\emph{reflected} (i.e., integrated) into object-level programming by
the use of quote and eval.

Metaprogramming with reflection is a very powerful programming tool.
Besides Lisp, several other programming languages employ quotation and
evaluation mechanisms to enable metaprogramming with reflection.
Examples include Agda~\cite{Norell07,Norell09}, Archon~\cite{Stump09},
Elixir~\cite{Elixir15}, F\#~\cite{FSharp15},
MetaML~\cite{TahaSheard00}, MetaOCaml~\cite{MetaOCaml11},
reFLect~\cite{GrundyEtAl06}, and Template
Haskell~\cite{SheardJones02}.

\bsp Analogous to metaprogramming in a programming language,
\emph{metareasoning} is performed in a logic by manipulating
meta-level values (e.g., syntax trees) that represent expressions in
the logic and is \emph{reflected} into object-level reasoning using
quotation and evaluation\footnote{Evaluation in this context is also
  called unquoting, interpretation, dereferencing, and dereification.}
mechanisms~\cite{Costantini02}.  In proof assistants like Coq and
Agda, metareasoning with reflection is implemented in the logic by
defining an infrastructure consisting of (1) an \emph{inductive type
  of syntactic values} that represent certain object-level
expressions, (2) an \emph{informal quotation operator} that maps these
object-level expressions to syntactic values, and (3) a \emph{formal
  evaluation operator} that maps syntactic values to the values of the
object-level expressions that they
represent~\cite{Chlipala13,GonthierEtAl15,VanDerWaltSwierstra12}.
Metareasoning with reflection is used for formalizing metalogical
techniques and incorporating symbolic computation into proof
assistants~\cite{Chlipala13,Farmer13,GonthierEtAl15,Harrison95,VanDerWaltSwierstra12}.
\esp

The metareasoning and reflection infrastructures that have been
employed in today's proof assistants are \emph{local} in the
sense that the syntactic values of the inductive type represent only a
subset of the expressions of the logic, the quotation operator can
only be applied to these expressions, and the evaluation operator can
only be applied to the syntactic values of the inductive type.  Can
metareasoning with reflection be implemented in a traditional logic
like first-order logic or simple type theory using a \emph{global}
infrastructure with quotation and evaluation operators like Lisp's
quote and eval?  This is largely an open question.  As far as we know,
there is no readily implementable version of a traditional logic that
admits global quotation and evaluation.  We have proposed a version of
NBG set theory named Chiron~\cite{FarmerArxiv13} and a version of
Alonzo Church's type theory~\cite{Church40}\footnote{Church's type
  theory is a version of simple type theory with lambda notation.}
named {\qzerouqe}~\cite{FarmerArxiv14} that include global quotation
and evaluation operators, but these logics have a high level of
complexity and are not easy to implement.

Many challenging problems face the logic engineer who seeks to
incorporate global quotation and evaluation into a traditional logic.
The three problems that most concern us are the following.  We will
write the quotation and evaluation operators applied to an expression
$e$ as $\synbrack{e}$ and $\sembrack{e}$, respectively.

\be

  \item \emph{Evaluation Problem.}  An evaluation operator is
    applicable to syntactic values that represent formulas and thus is
    effectively a truth predicate.  Hence, by the proof of Alfred
    Tarski's theorem on the undefinability of truth~\cite{Tarski35a},
    if the evaluation operator is total in the context of a
    sufficiently strong theory like first-order Peano arithmetic, then
    it is possible to express the liar paradox using the quotation and
    evaluation operators.  Therefore, the evaluation operator must be
    partial and the law of disquotation cannot hold universally (i.e.,
    for some expressions $e$, $\sembrack{\synbrack{e}} \not= e$).  As
    a result, reasoning with evaluation is cumbersome and leads to
    undefined expressions.

\iffalse
Tarski33,Tarski35
\fi

  \item \emph{Variable Problem.}  The variable $x$ is not free in the
    expression $\synbrack{x + 3}$ (or in any quotation).  However, $x$
    is free in $\sembrack{\synbrack{x + 3}}$ because
    $\sembrack{\synbrack{x + 3}} = x + 3$.  If the value of a constant
    $c$ is $\synbrack{x + 3}$, then $x$ is free in $\sembrack{c}$
    because $\sembrack{c} = \sembrack{\synbrack{x + 3}} = x + 3$.
    Hence, in the presence of an evaluation operator, whether or not a
    variable is free in an expression may depend on the values of the
    expression's components.  As a consequence, the substitution of an
    expression for the free occurrences of a variable in another
    expression depends on the semantics (as well as the syntax) of the
    expressions involved and must be integrated with the proof system
    of the logic.  That is, a logic with quotation and evaluation
    requires a semantics-dependent form of substitution in which side
    conditions, like whether a variable is free in an expression, are
    proved within the proof system.  This is a major departure from
    traditional logic.

  \item \emph{Double Substitution Problem.}  By the semantics of
    evaluation, the value of $\sembrack{e}$ is the \emph{value} of the
    expression whose syntax tree is represented by the \emph{value} of
    $e$.  Hence the semantics of evaluation involves a double
    valuation (see condition 6 of the definition of a model in
    section~\ref{subsec:models}).  If the value of a variable $x$ is
    $\synbrack{x}$, then $\sembrack{x} = \sembrack{\synbrack{x}} = x =
    \synbrack{x}$.  Hence the substitution of $\synbrack{x}$ for $x$
    in $\sembrack{x}$ requires one substitution inside the argument of
    the evaluation operator and another substitution after the
    evaluation operator is eliminated.  This double substitution is
    another major departure from traditional logic.

\ee

{\churchqe} is a version of Church's type theory~\cite{Church40} with
quotation and evaluation that overcomes these three problems.  It is
much simpler than {\qzerouqe} since (1) the quotation operator can
only be applied to expressions that do not contain the evaluation
operator and (2) substitution is not a logical constant (applied to
syntactic values).  Like {\qzerouqe}, {\churchqe} is based on
{\qzero}~\cite{Andrews02}, Peter Andrews' version of Church's type
theory.  In this paper, we present the syntax and semantics of
{\churchqe} and give several examples that illustrate the usefulness
of having quotation and evaluation in {\churchqe}.  We do not give a
proof system for {\churchqe}, but we do sketch what a proof system
could look like.

\section{Syntax}

The syntax of {\churchqe} is very similar to the syntax of
{\qzero}~\cite[pp.~210--211]{Andrews02}.  {\churchqe} has the syntax
of Church's type theory plus an inductive type of syntactic values, a
quotation operator, and a typed evaluation operator.  Like {\qzero},
the propositional connectives and quantifiers are defined using
function application, function abstraction, and equality.  For the
sake of simplicity, {\churchqe} does not contain, as in {\qzero}, a
definite description operator or, as in the logic of
HOL~\cite{GordonMelham93}, an indefinite description (choice)
operator or type variables.

\subsection{Types}

A \emph{type} of {\churchqe} is a string of symbols defined inductively by
the following formation rules:
\be

  \item \emph{Type of individuals}: $\iota$ is a type.

  \item \emph{Type of truth values}: $\omicron$ is a type.

  \item \emph{Type of constructions}: $\epsilon$ is a type.

  \item \emph{Function type}: If $\alpha$ and $\beta$ are types, then
    $(\alpha \tarrow \beta)$ is a type.\footnote{In Andrews'
    {\qzero}~\cite{Andrews02} and Church's original
    system~\cite{Church40}, the function type $(\alpha \tarrow \beta)$
    is written as $(\beta\alpha)$.}

\ee
Let $\sT$ denote the set of types of {\churchqe}.
$\alpha,\beta,\gamma, \ldots$ are syntactic variables ranging over
types.  When there is no loss of meaning, matching pairs of
parentheses in types may be omitted.  We assume that function type
formation associates to the right so that a type of the form $(\alpha
\tarrow (\beta \tarrow \gamma)) $ may be written as $\alpha \tarrow
\beta \tarrow \gamma$.

We will see in the next section that in {\churchqe} types are
directly assigned to variables and constants and thereby indirectly
assigned to expressions.

\subsection{Expressions}

A \emph{typed symbol} is a symbol with a subscript from $\sT$.  Let
$\sV$ be a set of typed symbols such that, for each $\alpha \in \sT$,
$\sV$ contains denumerably many typed symbols with subscript~$\alpha$.
A \emph{variable of type $\alpha$} of {\churchqe} is a member of $\sV$
with subscript $\alpha$.  $\textbf{f}_\alpha, \textbf{g}_\alpha,
\textbf{h}_\alpha, \textbf{u}_\alpha, \textbf{v}_\alpha,
\textbf{w}_\alpha,\textbf{x}_\alpha, \textbf{y}_\alpha,
\textbf{z}_\alpha,\ldots$ are syntactic variables ranging over
variables of type~$\alpha$.  We will assume that $f_\alpha, g_\alpha,
h_\alpha, u_\alpha, v_\alpha, w_\alpha, x_\alpha, y_\alpha,
z_\alpha,\ldots$ are actual variables of type~$\alpha$ of {\churchqe}.

\iffalse
Let $\sC$ be a set of typed symbols disjoint from $\sV$.  A
\emph{constant of type $\alpha$} of {\churchqe} is a member of $\sC$
with subscript~$\alpha$.  $\sC$ includes the following \emph{logical
  constants} of {\churchqe}: $\mname{=}_{\alpha \tarrow \alpha \tarrow
  o}$, $\mname{is-var}_{\epsilon \tarrow o}$,
$\mname{is-con}_{\epsilon \tarrow o}$, $\mname{app}_{\epsilon \tarrow
  \epsilon \tarrow \epsilon}$, $\mname{abs}_{\epsilon \tarrow \epsilon
  \tarrow \epsilon}$, $\mname{quo}_{\epsilon \tarrow \epsilon}$,
$\mname{is-expr}_{\epsilon \tarrow o}^{\alpha}$ for all $\alpha \in
\sT$.  $\textbf{c}_\alpha$ is a syntactic variable ranging over
constants of type~$\alpha$.
\fi

Let $\sC$ be a set of typed symbols disjoint from $\sV$ that includes
the typed symbols in Table~\ref{tab:log-con}.  A \emph{constant of
  type~$\alpha$} of {\churchqe} is a member of $\sC$ with
subscript~$\alpha$.  The typed symbols in Table~\ref{tab:log-con} are
the \emph{logical constants} of {\churchqe}.  $\textbf{c}_\alpha,
\textbf{d}_\alpha, \ldots$ are syntactic variables ranging over
constants of type~$\alpha$.

\begin{table}
\bc
\begin{tabular}{|ll|}
\hline
$\mname{=}_{\alpha \tarrow \alpha \tarrow o}$ 
& for all $\alpha \in \sT$\\
$\mname{is-var}_{\epsilon \tarrow o}$
&\\
$\mname{is-con}_{\epsilon \tarrow o}$
&\\
$\mname{app}_{\epsilon \tarrow \epsilon \tarrow \epsilon}$
&\\
$\mname{abs}_{\epsilon \tarrow \epsilon \tarrow \epsilon}$
&\\
$\mname{quo}_{\epsilon \tarrow \epsilon}$
&\\
$\mname{is-expr}_{\epsilon \tarrow o}^{\alpha}$
& for all $\alpha \in \sT$\\
\hline
\end{tabular}
\ec
\caption{Logical Constants}\label{tab:log-con}
\end{table}

An \emph{expression of type $\alpha$} of {\churchqe} is a string of
symbols defined inductively by the formation rules below.
$\textbf{A}_\alpha, \textbf{B}_\alpha, \textbf{C}_\alpha, \ldots$ are
syntactic variables ranging over expressions of type $\alpha$.  An
expression is \emph{eval-free} if it is constructed using just the
first five formation rules.
\be

  \item \emph{Variable}: $\textbf{x}_\alpha$ is an expression of type
    $\alpha$.

  \item \emph{Constant}: $\textbf{c}_\alpha$ is an expression of type
    $\alpha$.

  \item \emph{Function application}: $(\textbf{F}_{\alpha \tarrow
    \beta} \, \textbf{A}_\alpha)$ is an expression of type $\beta$.

  \item \emph{Function abstraction}: $(\LambdaApp \textbf{x}_\alpha
    \mdot \textbf{B}_\beta)$ is an expression of type $\alpha \tarrow
    \beta$.

  \item \emph{Quotation}: $\synbrack{\textbf{A}_\alpha}$ is an
    expression of type $\epsilon$ if $\textbf{A}_\alpha$ is eval-free.

  \item \emph{Evaluation}: $\sembrack{\textbf{A}_\epsilon}_{{\bf
      B}_\beta}$ is an expression of type $\beta$.

\ee 

\noindent
The purpose of the second component $\textbf{B}_\beta$ in an
evaluation $\sembrack{\textbf{A}_\epsilon}_{{\bf B}_\beta}$ is to
establish the type of the evaluation.  A \emph{formula} is an
expression of type $o$.  When there is no loss of meaning, matching
pairs of parentheses in expressions may be omitted.  We assume that
function application formation associates to the left so that an
expression of the form $((\textbf{G}_{\alpha \tarrow \beta \tarrow
  \gamma} \, \textbf{A}_\alpha) \, \textbf{B}_\beta)$ may be written
as $\textbf{G}_{\alpha \tarrow \beta \tarrow \gamma} \,
\textbf{A}_\alpha \, \textbf{B}_\beta$.

\subsection{Constructions}

A \emph{construction} of {\churchqe} is an expression of type
$\epsilon$ defined inductively as follows:

\be

  \item $\synbrack{\textbf{x}_\alpha}$ is a construction.

  \item $\synbrack{\textbf{c}_\alpha}$ is a construction.

  \item If $\textbf{A}_\epsilon$ and $\textbf{B}_\epsilon$ are
    constructions, then $\mname{app}_{\epsilon \tarrow \epsilon
      \tarrow \epsilon} \, \textbf{A}_\epsilon \,
    \textbf{B}_\epsilon$, $\mname{abs}_{\epsilon \tarrow \epsilon
      \tarrow \epsilon} \, \textbf{A}_\epsilon \,
    \textbf{B}_\epsilon$, and $\mname{quo}_{\epsilon \tarrow \epsilon}
    \, \textbf{A}_\epsilon$ are constructions.

\ee

\noindent
The set of constructions is thus an inductive type whose base elements
are quotations of variables and constants and whose constructors are
$\mname{app}_{\epsilon \tarrow \epsilon \tarrow \epsilon}$,
$\mname{abs}_{\epsilon \tarrow \epsilon \tarrow \epsilon}$, and
$\mname{quo}_{\epsilon \tarrow \epsilon}$.  We will call these three
constants \emph{syntax constructors}.

Let $\sE$ be the function mapping eval-free expressions to
constructions that is defined inductively as follows:

\be

  \item $\sE(\textbf{x}_\alpha) = \synbrack{\textbf{x}_\alpha}$.

  \item $\sE(\textbf{c}_\alpha) = \synbrack{\textbf{c}_\alpha}$.

  \item $\sE(\textbf{F}_{\alpha \tarrow \beta} \, \textbf{A}_\alpha) =
    \mname{app}_{\epsilon \tarrow \epsilon \tarrow \epsilon} \,
    \sE(\textbf{F}_{\alpha \tarrow \beta}) \, \sE(\textbf{A}_\alpha)$.

  \item $\sE(\LambdaApp \textbf{x}_\alpha \mdot \textbf{B}_\beta) =
    \mname{abs}_{\epsilon \tarrow \epsilon \tarrow \epsilon} \,
    \sE(\textbf{x}_\alpha) \, \sE(\textbf{B}_\beta)$.

  \item $\sE(\synbrack{\textbf{A}_\alpha}) = \mname{quo}_{\epsilon
    \tarrow \epsilon} \, \sE(\textbf{A}_\alpha)$.

\ee

\noindent
$\sE$ is clearly injective.  When $\textbf{A}_\alpha$ is eval-free,
$\sE(\textbf{A}_\alpha)$ is a construction that represents the syntax
tree of $\textbf{A}_\alpha$.  That is, $\sE(\textbf{A}_\alpha)$ is a
syntactic value that represents how $\textbf{A}_\alpha$ is
syntactically constructed.  For every eval-free expression, there is a
construction that represents its syntax tree, but not every
construction represents the syntax tree of an eval-free expression.
For example, $\mname{app}_{\epsilon \tarrow \epsilon \tarrow \epsilon}
\, \synbrack{\textbf{x}_\alpha} \, \synbrack{\textbf{x}_\alpha}$
represents the syntax tree of $(\textbf{x}_\alpha \,
\textbf{x}_\alpha)$ which is not an expression of {\churchqe} since
the types are mismatched.  A construction is \emph{proper} if it is in
the range of $\sE$, i.e., it represents the syntax tree of an
eval-free expression.

The five kinds of eval-free expressions and the syntactic values that
represent their syntax trees are given in
Table~\ref{tab:eval-free-exprs}.  

\iffalse 

The constants $\mname{is-var}_{\epsilon \tarrow o}$,
$\mname{is-con}_{\epsilon \tarrow o}$, and $\mname{is-expr}_{\epsilon
  \tarrow o}^{\alpha}$ are used to assert that a construction
represents the syntax tree of a variable, constant, and expression of
type $\alpha$, respectively.

The five kinds of eval-free expressions and the syntactic values that
represent their syntax trees are given in
Table~\ref{tab:eval-free-exprs}.  We will see in the next section that
$\mname{is-var}_{\epsilon \tarrow o} \, \textbf{A}_\epsilon$,
$\mname{is-con}_{\epsilon \tarrow o} \, \textbf{A}_\epsilon$, and
$\mname{is-expr}_{\epsilon \tarrow o}^{\alpha} \, \textbf{A}_\epsilon$
assert that $\textbf{A}_\epsilon$ represents the syntax tree of a
variable, constant, and expression of type $\alpha$, respectively.

\fi

\begin{table}[b]
\bc
\begin{tabular}{|lll|}
\hline

\textbf{Kind}
& \textbf{Syntax}
& \textbf{Syntactic Values}\\

Variable \hspace*{15ex}
& $\textbf{x}_\alpha$  \hspace*{9ex}
& $\synbrack{\textbf{x}_\alpha}$\\

Constant
& $\textbf{c}_\alpha$
& $\synbrack{\textbf{c}_\alpha}$\\

Function application
& $\textbf{F}_{\alpha \tarrow \beta} \, \textbf{A}_\alpha$
& $\mname{app}_{\epsilon \tarrow \epsilon \tarrow \epsilon} \,
  \sE(\textbf{F}_{\alpha \tarrow \beta}) \, \sE(\textbf{A}_\alpha)$\\

Function abstraction
& $\LambdaApp \textbf{x}_\alpha \mdot \textbf{B}_\beta$
& $\mname{abs}_{\epsilon \tarrow \epsilon \tarrow \epsilon} \,
  \sE(\textbf{x}_\alpha) \, \sE(\textbf{B}_\beta)$\\

Quotation
& $\synbrack{\textbf{A}_\alpha}$
& $\mname{quo}_{\epsilon \tarrow \epsilon} \, \sE(\textbf{A}_\alpha)$\\

\hline
\end{tabular}
\ec
\caption{Five Kinds of Eval-Free Expressions}\label{tab:eval-free-exprs}
\end{table}

\subsection{Definitions and Abbreviations} \label{subsec:definitions}

As Andrews does in~\cite[p.~212]{Andrews02}, we introduce in
Table~\ref{tab:defs} several defined logical constants and
abbreviations.  The former includes constants for true and false and
the propositional connectives.  The latter includes notation for
equality, the propositional connectives, universal and existential
quantification, and a simplified notation for evaluations.

\begin{table}
\bc
\begin{tabular}{|lll|}
\hline

$(\textbf{A}_\alpha = \textbf{B}_\alpha)$ \hspace{1ex}
& stands for  \hspace{1ex}
& $=_{\alpha \tarrow \alpha \tarrow o} \, \textbf{A}_\alpha \, \textbf{B}_\alpha$.\\

%% $(\textbf{A}_o \Iff \textbf{B}_o)$ 
%% & stands for 
%% & $=_{o \tarrow o \tarrow o} \, \textbf{A}_o \, \textbf{B}_o$.\\

$T_o$ 
& stands for
& $=_{o \tarrow o \tarrow o} \; = \; =_{o \tarrow o \tarrow o}$.\\

$F_o$ 
& stands for
& $(\LambdaApp x_o \mdot T_o) = (\LambdaApp x_o \mdot x_o).$\\

\iffalse
& $(\LambdaApp \textbf{x}_o \mdot T_o) = (\LambdaApp \textbf{x}_o \mdot \textbf{x}_o)$\\
&
&\mbox{for some chosen $\textbf{x}_o \in \sV$}.\\
\fi

$(\ForallApp \textbf{x}_\alpha \mdot \textbf{A}_o)$ 
& stands for
& $(\LambdaApp \textbf{x}_\alpha \mdot T_o) = (\LambdaApp \textbf{x}_\alpha \mdot \textbf{A}_o)$.\\

$\wedge_{o \tarrow o \tarrow o}$ 
& stands for
& $\LambdaApp x_o \mdot \LambdaApp y_o \mdot {}$\\
& 
& \hspace{2ex}$((\LambdaApp g_{o \tarrow o \tarrow o} \mdot 
g_{o \tarrow o \tarrow o} \, T_o \, T_o) = {}$\\
&
& \hspace{3ex}$(\LambdaApp g_{o \tarrow o \tarrow o} \mdot 
g_{o \tarrow o \tarrow o} \, x_o \, y_o)).$\\

\iffalse
& $\LambdaApp \textbf{x}_o \mdot \LambdaApp \textbf{y}_o \mdot {}$\\
& 
& \hspace{2ex}$((\LambdaApp \textbf{g}_{o \tarrow o \tarrow o} \mdot 
\textbf{g}_{o \tarrow o \tarrow o} \, T_o \, T_o) = {}$\\
&
& \hspace{3ex}$(\LambdaApp \textbf{g}_{o \tarrow o \tarrow o} \mdot 
\textbf{g}_{o \tarrow o \tarrow o} \, \textbf{x}_o \, \textbf{y}_o))$\\
&
&\mbox{for some chosen $\textbf{x}_o, \textbf{y}_o,  \textbf{g}_{o \tarrow o \tarrow o}\in \sV$}.\\
\fi

$(\textbf{A}_o \Andd \textbf{B}_o)$ 
& stands for
& $\wedge_{o \tarrow o \tarrow o} \, \textbf{A}_o \, \textbf{B}_o$.\\

$\Implies_{o \tarrow o \tarrow o}$ 
& stands for
& $\LambdaApp x_o \mdot \LambdaApp y_o \mdot (x_o = (x_o \Andd y_o)).$\\ 

\iffalse
& $\LambdaApp \textbf{x}_o \mdot \LambdaApp \textbf{y}_o \mdot (\textbf{x}_o = (\textbf{x}_o \Andd \textbf{y}_o))$\\ 
&
&\mbox{for some chosen $\textbf{x}_o, \textbf{y}_o \in \sV$}.\\
\fi

$(\textbf{A}_o \Implies \textbf{B}_o)$ 
& stands for
& ${\Implies_{o \tarrow o \tarrow o}} \, \textbf{A}_o \,\textbf{B}_o$.\\

$\Neg_{o \tarrow o}$ 
& stands for
& ${=_{o \tarrow o \tarrow o}} \, F_o$.\\

$(\Neg\textbf{A}_o)$ 
& stands for
& $\Neg_{o \tarrow o} \, \textbf{A}_o$.\\

$\vee_{o \tarrow o \tarrow o}$ 
& stands for
& $\LambdaApp x_o \mdot \LambdaApp y_o \mdot \Neg (\Neg x_o \Andd \Neg y_o).$\\

\iffalse
& $\LambdaApp \textbf{x}_o \mdot \LambdaApp \textbf{y}_o \mdot \Neg (\Neg\textbf{x}_o \Andd \Neg \textbf{y}_o)$\\
&
&\mbox{for some chosen $\textbf{x}_o, \textbf{y}_o \in \sV$}.\\
\fi

$(\textbf{A}_o \Or \textbf{B}_o)$ 
& stands for
& ${\vee_{o \tarrow o \tarrow o}} \, \textbf{A}_o \, \textbf{B}_o$.\\

$(\ForsomeApp \textbf{x}_\alpha \mdot \textbf{A}_o)$ 
& stands for
& $\Neg(\ForallApp \textbf{x}_\alpha \mdot \Neg\textbf{A}_o)$.\\

%% $(\textbf{A}_\alpha \not= \textbf{B}_\alpha)$ 
%% & stands for 
%% & $\Neg(\textbf{A}_\alpha = \textbf{B}_\alpha)$.\\

$\sembrack{\textbf{A}_\epsilon}_\beta$ 
& stands for
& $\sembrack{\textbf{A}_\epsilon}_{{\bf B}_\beta}$.\\

\hline
\end{tabular}
\ec
\caption{Definitions and Abbreviations}\label{tab:defs}
\end{table}

\section{Semantics}

The semantics of {\churchqe} extends the semantics of
{\qzero}~\cite[pp.~238--239]{Andrews02} by defining the domain of the
type $\epsilon$ and what quotations and evaluations mean.

\subsection{Frames}

A \emph{frame} of {\churchqe} is a collection $\set{D_\alpha \;|\;
  \alpha \in \sT}$ of domains such that:

\be

  \item $D_\iota$ is a nonempty set of values (called \emph{individuals}).

  \item $D_o = \set{\TRUE,\FALSE}$, the set of standard \emph{truth
    values}.

  \item $D_\epsilon$ is the set of \emph{constructions} of {\churchqe}.

  \item For $\alpha, \beta \in \sT$, $D_{\alpha \rightarrow \beta}$ is the set
    of \emph{total functions} from $D_\alpha$ to $D_\beta$.

\ee
%
\iffalse
$\sD_\iota$ is the \emph{domain of individuals}, $\sD_o$ is the
\emph{domain of truth values}, $\sD_\epsilon$ is the \emph{domain of
  constructions}, and, for $\alpha, \beta \in \sT$,
$\sD_{(\alpha\beta)}$ is a \emph{function domain}.
\fi

\subsection{Interpretations}

An \emph{interpretation} of {\churchqe} is a pair $(\set{D_\alpha
  \;|\; \alpha \in \sT},I)$ consisting of a frame and an
interpretation function $I$ that maps each constant in $\sC$ of type
$\alpha$ to an element of $D_\alpha$ such that:

\be

  \item For all $\alpha \in \sT$, $I(=_{\alpha \rightarrow \alpha
    \rightarrow o})$ is the function $f \in D_{\alpha \rightarrow
    \alpha \rightarrow o}$ such that, for all $d_1,d_2 \in D_\alpha$,
    $f(d_1)(d_2) = \TRUE$ iff $d_1 = d_2$.  That is, $I(=_{\alpha
    \rightarrow \alpha \rightarrow o})$ is the identity relation on
    $D_\alpha$.

  \item $I(\mname{is-var}_{\epsilon \tarrow o})$ is the function $f
    \in D_{\epsilon \rightarrow o}$ such that, for all
    $\textbf{A}_\epsilon \in D_\epsilon$, $f(\textbf{A}_\epsilon) =
    \TRUE$ iff $\textbf{A}_\epsilon = \synbrack{\textbf{x}_\alpha}$
    for some variable $\textbf{x}_\alpha \in \sV$.

  \item $I(\mname{is-con}_{\epsilon \tarrow o})$ is the function $f
    \in D_{\epsilon \rightarrow o}$ such that, for all
    $\textbf{A}_\epsilon \in D_\epsilon$, $f(\textbf{A}_\epsilon) =
    \TRUE$ iff $\textbf{A}_\epsilon = \synbrack{\textbf{c}_\alpha}$
    for some constant $\textbf{c}_\alpha \in \sC$.

  \item \bsp $I(\mname{app}_{\epsilon \tarrow \epsilon \tarrow
    \epsilon})$ is the function $f \in D_{\epsilon \tarrow \epsilon
    \tarrow \epsilon}$ such that, for all $\textbf{A}_\epsilon,
    \textbf{B}_\epsilon \in D_\epsilon$,
    $f(\textbf{A}_\epsilon)(\textbf{B}_\epsilon)$ is the construction
    $\mname{app}_{\epsilon \tarrow \epsilon \tarrow \epsilon} \,
    \textbf{A}_\epsilon \, \textbf{B}_\epsilon$. \esp

  \item $I(\mname{abs}_{\epsilon \tarrow \epsilon \tarrow \epsilon})$
    is the function $f \in D_{\epsilon \tarrow \epsilon \tarrow
      \epsilon}$ such that, for all $\textbf{A}_\epsilon,
    \textbf{B}_\epsilon \in D_\epsilon$,
    $f(\textbf{A}_\epsilon)(\textbf{B}_\epsilon)$ is the construction
    $\mname{abs}_{\epsilon \tarrow \epsilon \tarrow \epsilon} \,
    \textbf{A}_\epsilon \, \textbf{B}_\epsilon$.

  \item $I(\mname{quo}_{\epsilon \tarrow \epsilon})$ is the function
    $f \in D_{\epsilon \tarrow \epsilon}$ such that, for all
    $\textbf{A}_\epsilon \in D_\epsilon$, $f(\textbf{A}_\epsilon)$ is
    the construction $\mname{quo}_{\epsilon \tarrow \epsilon} \,
    \textbf{A}_\epsilon$.

  \item For all $\alpha \in \sT$, $I(\mname{is-expr}_{\epsilon \tarrow
    o}^{\alpha})$ is the function $f \in D_{\epsilon \rightarrow o}$
    such that, for all $\textbf{A}_\epsilon \in D_\epsilon$,
    $f(\textbf{A}_\epsilon) = \TRUE$ iff $\textbf{A}_\epsilon =
    \sE(\textbf{B}_\alpha)$ for some (eval-free) expression
    $\textbf{B}_\alpha$.

\ee

\begin{rem}[Domain of Constructions]\em
We would prefer to define $D_\epsilon$ to be the set of proper
constructions because we need only proper constructions to represent
the syntax trees of eval-free expressions.  However, then the natural
interpretations of the three syntax constructors ---
$\mname{app}_{\epsilon \tarrow \epsilon \tarrow \epsilon}$,
$\mname{abs}_{\epsilon \tarrow \epsilon \tarrow \epsilon}$, and
$\mname{quo}_{\epsilon \tarrow \epsilon}$ --- would be partial
functions.  Since {\churchqe} admits only total functions, it is more
convenient to allow $D_\epsilon$ to include improper constructions
than to interpret the syntax constructors as total functions that
represent partial functions.
\end{rem}

An \emph{assignment} into a frame $\set{D_\alpha \;|\; \alpha \in
  \sT}$ is a function $\phi$ whose domain is $\sV$ such that, for each
variable $\textbf{x}_\alpha$, $\phi(\textbf{x}_\alpha) \in D_\alpha$.
Given an assignment $\phi$, a variable $\textbf{x}_\alpha$, and $d \in
D_\alpha$, let $\phi[\textbf{x}_\alpha \mapsto d]$ be the assignment
$\psi$ such that $\psi(\textbf{x}_\alpha) = d$ and
$\psi(\textbf{y}_\beta) = \phi(\textbf{y}_\beta)$ for all variables
$\textbf{y}_\beta \not= \textbf{x}_\alpha$.  Given an interpretation
$\sM = (\set{D_\alpha \;|\; \alpha \in \sT}, I)$,
$\mname{assign}(\sM)$ is the set of assignments into the frame of
$\sM$.

\subsection{Models} \label{subsec:models}

An interpretation $\sM = (\set{D_\alpha \;|\; \alpha \in \sT), I}$ is
a \emph{model} for {\churchqe} if there is a binary valuation
  function $V^{\cal M}$ such that, for all assignments $\phi \in
  \mname{assign}(\sM)$ and expressions $\textbf{C}_\gamma$, $V^{\cal
    M}_{\phi}(\textbf{C}_\gamma) \in D_\gamma$ and each of the
  following conditions is satisfied:

\be

  \item If $\textbf{C}_\gamma \in \sV$, then $V^{\cal
    M}_{\phi}(\textbf{C}_\gamma) = \phi(\textbf{C}_\gamma)$.

  \item If $\textbf{C}_\gamma \in \sC$, then $V^{\cal
    M}_{\phi}(\textbf{C}_\gamma) = I(\textbf{C}_\gamma)$.

  \item If $\textbf{C}_\gamma$ is $\textbf{F}_{\alpha \tarrow \beta} \,
    \textbf{A}_\alpha$, then $V^{\cal M}_{\phi}(\textbf{C}_\gamma) =
    V^{\cal M}_{\phi}(\textbf{F}_{\alpha \tarrow \beta})(V^{\cal
      M}_{\phi}(\textbf{A}_\alpha))$.

  \item If $\textbf{C}_\gamma$ is $\LambdaApp \textbf{x}_\alpha \mdot
    \textbf{B}_\beta$, then $V^{\cal M}_{\phi}(\textbf{C}_\gamma)$ is
    the function $f \in D_{\alpha \tarrow \beta}$ such that, for each
    $d \in D_\alpha$, $f(d) = V^{\cal M}_{\phi[{\bf x}_\alpha \mapsto
      d]}(\textbf{B}_\beta)$.

  \item If $\textbf{C}_\gamma$ is $\synbrack{\textbf{A}_\alpha}$, then
    $V^{\cal M}_{\phi}(\textbf{C}_\gamma) = \sE(\textbf{A}_\alpha)$.

  \item If $\textbf{C}_\gamma$ is
    $\sembrack{\textbf{A}_\epsilon}_\beta$ and $V^{\cal
    M}_{\phi}(\mname{is-expr}_{\epsilon \tarrow o}^{\beta} \,
    \textbf{A}_\epsilon) = \TRUE$, then \[V^{\cal
      M}_{\phi}(\textbf{C}_\gamma) = V^{\cal
      M}_{\phi}(\sE^{-1}(V^{\cal M}_{\phi}(\textbf{A}_\epsilon))).\]

\ee 

\begin{prop}
Models for {\churchqe} exist.
\end{prop}

\begin{proof} \bsp
It is easy to construct an interpretation $\sM = (\set{\sD_\alpha
  \;|\; \alpha \in \sT}, I)$ that is a model for {\churchqe}.  Note
that, if $V^{\cal M}_{\phi}(\mname{is-expr}_{\epsilon \tarrow
  o}^{\beta} \, \textbf{A}_\epsilon) = \FALSE$, then $V^{\cal
  M}_{\phi}(\sembrack{\textbf{A}_\epsilon}_\beta)$ can be any value in
$D_\beta$. \hfill $\Box$\esp
\end{proof}

\begin{rem}[Standard vs.~General Models]\em
The notion of a model defined here is a \emph{standard model} in which
each function domain $D_{\alpha \rightarrow \beta}$ is the set of
\emph{all} total functions from $D_\alpha$ to $D_\beta$.  Andrews'
semantics for {\qzero} is based on the notion of a \emph{general
  model}, introduced by Leon Henkin~\cite{Henkin50}, in which each
function domain $D_{\alpha \rightarrow \beta}$ is a set of \emph{some}
total functions from $D_\alpha$ to $D_\beta$.  General models can be
easily defined for {\churchqe}.  The definition of a frame, however,
has to be changed so that the domain $D_\epsilon$ may include
``nonstandard constructions''.
\end{rem}

\begin{rem}[Semantics of Evaluations]\em\bsp
When $V^{\cal M}_{\phi}(\mname{is-expr}_{\epsilon \tarrow o}^{\beta}
\, \textbf{A}_\epsilon) = \TRUE$, the semantics of $V^{\cal
  M}_{\phi}(\sembrack{\textbf{A}_\epsilon}_\beta)$ involves a double
valuation as mentioned in the Double Substitution Problem described in
the Introduction.\esp
\end{rem}

\begin{rem}[Undefined Evaluations]\em
Suppose $V^{\cal M}_{\phi}(\textbf{A}_\epsilon)$ is an improper
construction.  Then $V^{\cal M}_{\phi}(\sE^{-1}(V^{\cal
  M}_{\phi}(\textbf{A}_\epsilon)))$ is undefined and $V^{\cal
  M}_{\phi}(\sembrack{\textbf{A}_\epsilon}_\beta)$ has no natural
value.  Since {\churchqe} does not admit undefined expressions,
$V^{\cal M}_{\phi}(\sembrack{\textbf{A}_\epsilon}_\beta)$ is defined
but its value is unspecified. Similarly, if $V^{\cal
  M}_{\phi}(\textbf{A}_\epsilon)$ is a proper construction of the form
$\sE(\textbf{B}_\gamma)$ with $\gamma \not= \beta$, $V^{\cal
  M}_{\phi}(\sembrack{\textbf{A}_\epsilon}_\beta)$ is unspecified.
\end{rem}

Let $\sM$ be a model for {\churchqe}. $\textbf{A}_o$ is \emph{valid}
in $\sM$, written $\sM \models \textbf{A}_o$, if $\sV^{\cal
  M}_{\phi}(\textbf{A}_o) = \mname{T}$ for all assignments $\phi \in
\mname{assign}(\sM)$.

\begin{prop} \label{prop:val-const}
Let $\sM$ be a model for {\churchqe}, $\textbf{A}_\epsilon$ be a
construction, and $\phi \in \mname{assign}(\sM)$.  Then $\sV^{\cal
  M}_{\phi}(\textbf{A}_\epsilon) = \textbf{A}_\epsilon$.
\end{prop}

\begin{proof}
Follows immediately from conditions 4--6 of the definition of an
interpretation and condition 5 of the definition of a model.\hfill $\Box$
\end{proof}

\begin{thm}[Law of Quotation] \label{thm:quotation}
$\synbrack{\textbf{A}_\alpha} = \sE(\textbf{A}_\alpha)$ is valid in
  every model of {\churchqe}.
\end{thm}

\begin{proof}
Let $\sM$ be a model of {\churchqe} and $\phi \in
\mname{assign}(\sM)$.  Then 
\begin{align} \setcounter{equation}{0}
&
\sV^{\cal M}_{\phi}(\synbrack{\textbf{A}_\alpha}) \\
&=
\sE(\textbf{A}_\alpha) \\
&=
\sV^{\cal M}_{\phi}(\sE(\textbf{A}_\alpha))
\end{align}
(2) follows from condition 5 of the definition of a model, and (3)
follows from Proposition~\ref{prop:val-const}. Hence $\sV^{\cal
  M}_{\phi}(\synbrack{\textbf{A}_\alpha}) = V^{\cal
  M}_{\phi}(\sE(\textbf{A}_\alpha))$ for all $\phi \in
\mname{assign}(\sM)$ which implies $\synbrack{\textbf{A}_\alpha} =
\sE(\textbf{A}_\alpha)$ is valid in $\sM$.\hfill $\Box$
\end{proof}

\begin{thm}[Law of Disquotation] \label{thm:disquotation}\bsp
$\sembrack{\synbrack{\textbf{A}_\alpha}}_\alpha = \textbf{A}_\alpha$
  is valid in every model of {\churchqe}.\esp
\end{thm}

\begin{proof}\bsp
Let $\sM$ be a model of {\churchqe} and $\phi \in
\mname{assign}(\sM)$.  Then
\begin{align} \setcounter{equation}{0}
&
\sV^{\cal M}_{\phi}(\sembrack{\synbrack{\textbf{A}_\alpha}}_\alpha) \\
&=
\sV^{\cal M}_{\phi}(\sE^{-1}(\sV^{\cal M}_{\phi}(\synbrack{\textbf{A}_\alpha}))) \\
&=
\sV^{\cal M}_{\phi}(\sE^{-1}(\sE(\textbf{A}_\alpha))) \\
&=
\sV^{\cal M}_{\phi}(\textbf{A}_\alpha)
\end{align}
Since $V^{\cal M}_{\phi}(\mname{is-expr}_{\epsilon \tarrow o}^{\alpha}
\, \synbrack{\textbf{A}_\alpha}) = \TRUE$, (2) follows from condition
6 of the definition of a model. $\sV^{\cal
  M}_{\phi}(\synbrack{\textbf{A}_\alpha}) = \sE(\textbf{A}_\alpha)$ by
condition 5 of the definition of a model.  (3) and (4) are then
immediate.  Hence $\sV^{\cal
  M}_{\phi}(\sembrack{\synbrack{\textbf{A}_\alpha}}_\alpha) =
\sV^{\cal M}_{\phi}(\textbf{A}_\alpha)$ for all $\phi \in
\mname{assign}(\sM)$ which implies
$\sembrack{\synbrack{\textbf{A}_\alpha}}_\alpha = \textbf{A}_\alpha$
is valid in $\sM$. \hfill $\Box$\esp
\end{proof}

\begin{rem}[Evaluation Problem]\em\bsp
Theorem~\ref{thm:disquotation} shows that disquotation holds
universally in {\churchqe} contrary to the Evaluation Problem
described in the Introduction.  We have avoided the Evaluation Problem
in {\churchqe} by admitting only quotations of eval-free expressions.
If quotations of non-eval-free expressions were allowed in
{\churchqe}, the logic would be significantly more expressive, but
also much more complicated, as seen in
{\qzerouqe}~\cite{FarmerArxiv14}.\esp
\end{rem}

\begin{rem}[Quotation restricted to Closed Expressions]\em \bsp
If quotation is restricted to closed eval-free expressions in
{\churchqe}, then the Variable Problem and Double Substitution Problem
disappear.  However, most of the usefulness of having quotation and
evaluation in {\churchqe} would also disappear --- which is
illustrated by the examples in the next section.\esp
\end{rem}

\section{Examples}

We will present in this section four examples that illustrate the
utility of the quotation and evaluation facility in {\churchqe}.

\subsection{Reasoning about Syntax}

Reasoning about the syntax of expressions is normally performed in the
metalogic, but in {\churchqe} reasoning about the syntax of eval-free
expressions can be performed in the logic itself.  This is done by
reasoning about constructions (which represent the syntax trees of
eval-free expressions) using quotation and the machinery of
constructions.  Algorithms that manipulate eval-free expressions can
be formalized as functions that manipulate constructions.  The
functions can be executed using beta-reduction, rewriting, and other
kinds of simplification.

As an example, consider the constant
$\mname{make-implication}_{\epsilon \tarrow \epsilon \tarrow
  \epsilon}$ defined as
\[\LambdaApp x_\epsilon \mdot \LambdaApp y_\epsilon \mdot
(\mname{app}_{\epsilon \tarrow \epsilon \tarrow \epsilon} \,
(\mname{app}_{\epsilon \tarrow \epsilon \tarrow \epsilon} \,
\synbrack{\Implies_{o \tarrow o \tarrow o}} \, x_\epsilon) \,
y_\epsilon).\] It can be used to build constructions that
represent implications.  As another example, consider the constant
$\mname{is-app}_{\epsilon \tarrow o}$ defined as
\[\LambdaApp x_\epsilon \mdot
\ForsomeApp y_\epsilon \mdot \ForsomeApp z_\epsilon
\mdot x_\epsilon = (\mname{app}_{\epsilon \tarrow \epsilon
  \tarrow \epsilon} \, y_\epsilon \, z_\epsilon).\]
It can be used to test whether a construction represents a function
application.

Reasoning about syntax is a two-step process: First, a construction is
built using quotation and the machinery of constructions, and second,
the construction is employed using evaluation.  Continuing the example
above, \[\mname{make-implication}_{\epsilon \tarrow \epsilon \tarrow
  \epsilon} \, \synbrack{\textbf{A}_o} \, \synbrack{\textbf{B}_o}\]
builds a construction equivalent to the quotation
$\synbrack{\textbf{A}_o \Implies \textbf{B}_o}$ and
\[\sembrack{\mname{make-implication}_{\epsilon \tarrow \epsilon \tarrow
  \epsilon} \, \synbrack{\textbf{A}_o} \, \synbrack{\textbf{B}_o}}_o\]
employs the construction as the implication $\textbf{A}_o \Implies
\textbf{B}_o$.  Using this mixture of quotation and evaluation, it is
possible to express the interplay of syntax and semantics that is
needed to formalize syntax-based algorithms that are commonly used in
mathematics~\cite{Farmer13}.  See
section~\ref{secsub:meaning-formulas} for an example.

\subsection{Quasiquotation}

Quasiquotation is a parameterized form of quotation in which the
parameters serve as holes in a quotation that are filled with
expressions that denote syntactic values.  It is a very powerful
syntactic device for specifying expressions and defining macros.
Quasiquotation was introduced by Willard Van Orman Quine in 1940 in
the first version of his book \emph{Mathematical
  Logic}~\cite{Quine03}.  It has been extensively employed in the Lisp
family of programming languages~\cite{Bawden99}.\footnote{In Lisp, the
  standard symbol for quasiquotation is the backquote (\texttt{`})
  symbol, and thus in Lisp, quasiquotation is usually called
  \emph{backquote}.}

In {\churchqe}, constructing a large quotation from smaller quotations
can be tedious because it requires many applications of syntax
constructors.  Quasiquotation provides a convenient way to construct
big quotations from little quotations.  It can be defined
straightforwardly in {\churchqe}.

A \emph{quasi-expression} of {\churchqe} is defined inductively as
follows:

\be

  \item $\commabrack{\textbf{A}_\epsilon}$ is a quasi-expression
    called an \emph{antiquotation}.

  \item $\textbf{x}_\alpha$ is a quasi-expression.

  \item $\textbf{c}_\alpha$ is a quasi-expression.

  \item If $M$ and $N$ are quasi-expressions, then $(M \, N)$,
    $(\LambdaApp \textbf{x}_\alpha \mdot N)$, $(\LambdaApp
    \commabrack{\textbf{A}_\epsilon} \mdot N)$, and $\synbrack{M}$ are
    quasi-expressions.

\ee

\noindent
A quasi-expression is thus an expression where one or more
subexpressions have been replaced by antiquotations.  For example,
$\Neg(\textbf{A}_o \Andd \commabrack{\textbf{B}_\epsilon})$ is a
quasi-expression.  Obviously, every expression is a quasi-expression.

Let $\sE'$ be the function mapping quasi-expressions to expressions
of type~$\epsilon$ that is defined inductively as follows:

\be

  \item $\sE'(\commabrack{\textbf{A}_\epsilon}) = \textbf{A}_\epsilon$.

  \item $\sE'(\textbf{x}_\alpha) = \synbrack{\textbf{x}_\alpha}$.

  \item $\sE'(\textbf{c}_\alpha) = \synbrack{\textbf{c}_\alpha}$.

  \item $\sE'(M \, N) = \mname{app}_{\epsilon \tarrow \epsilon \tarrow
    \epsilon} \, \sE'(M) \, \sE'(N)$.

  \item $\sE'(\LambdaApp M \mdot N) = \mname{abs}_{\epsilon \tarrow
    \epsilon \tarrow \epsilon} \, \sE'(M) \, \sE'(N)$.

  \item $\sE(\synbrack{M}) = \mname{quo}_{\epsilon \tarrow \epsilon} \,
    \sE'(M)$.

\ee

\noindent
Notice that $\sE'(M) = \sE(M)$ when $M$ is an expression.  Continuing
our example above, $\sE'(\Neg(\textbf{A}_o \Andd
\commabrack{\textbf{B}_\epsilon})) = {}$
\[\mname{app}_{\epsilon \tarrow
  \epsilon \tarrow \epsilon} \, \synbrack{\Neg_{o \tarrow o}} \,
(\mname{app}_{\epsilon \tarrow \epsilon \tarrow \epsilon} \,
(\mname{app}_{\epsilon \tarrow \epsilon \tarrow \epsilon}
\synbrack{\wedge_{o \tarrow o \tarrow o}} \, \sE'(\textbf{A}_o))
\, \textbf{B}_\epsilon).\]

A \emph{quasiquotation} is an expression of the form $\synbrack{M}$
where $M$ is a quasi-expression.  Thus every quotation is a
quasiquotation.  The quasiquotation $\synbrack{M}$ serves as an
alternate notation for the expression $\sE'(M)$.  So
$\synbrack{\Neg(\textbf{A}_o \Andd \commabrack{\textbf{B}_\epsilon})}$
stands for the significantly more verbose expression in the previous
paragraph.  It represents the syntax tree of a negated conjunction in
which the part of the tree corresponding to the second conjunct is
replaced by the syntax tree represented by $\textbf{B}_\epsilon$.  If
$\textbf{B}_\epsilon$ is a quotation $\synbrack{\textbf{C}_o}$, then
the quasiquotation $\synbrack{\Neg(\textbf{A}_o \Andd
  \commabrack{\synbrack{\textbf{C}_o}})}$ is equivalent to the
quotation $\synbrack{\Neg(\textbf{A}_o \Andd \textbf{C}_o)}$.

\subsection{Schemas}

A \emph{schema} is a metalogical expression containing syntactic
variables.  An instance of a schema is a logical expression obtained
by replacing the syntactic variables with appropriate logical
expressions.  In {\churchqe}, a schema can be formalized as a single
logical expression.

For example, consider the \emph{law of excluded middle (LEM)} that is
expressed as the formula schema $A \Or \Neg A$ where $A$ is a
syntactic variable ranging over all formulas.  LEM can be
formalized in {\churchqe} as the universal statement
\[\ForallApp x_\epsilon \mdot 
\mname{is-expr}_{\epsilon \tarrow o}^{o} \, x_\epsilon \Implies
\sembrack{x_\epsilon}_o \Or \Neg \sembrack{x_\epsilon}_o.\] An
instance of this formalization of LEM is any instance of the universal
statement.  Using quasiquotation, LEM could also be formalized in
{\churchqe} as
\[\ForallApp x_\epsilon \mdot 
\mname{is-expr}_{\epsilon \tarrow o}^{o} \, x_\epsilon \Implies
\sembrack{\synbrack{\commabrack{x_\epsilon} \Or \Neg
    \commabrack{x_\epsilon}}}_o.\]

If we assume that the domain of the type $\iota$ is the natural
numbers and $\sC$ includes the usual constants of natural number
arithmetic (including a constant $\mname{S}_{\iota \tarrow \iota}$
representing the successor function), then the (first-order)
\emph{induction schema for Peano arithmetic} can be formalized in
     {\churchqe} as
\begin{align*}
&
\ForallApp f_\epsilon \mdot 
(\mname{is-expr}_{\epsilon \tarrow o}^{\iota \tarrow o} \, f_\epsilon \Andd
\mname{is-peano}_{\epsilon \tarrow o} \, f_\epsilon) \Implies {} \\
&
\hspace{4ex}
((\sembrack{f_\epsilon}_{\iota \tarrow o} \, 0 \Andd
(\ForallApp x_\iota \mdot \sembrack{f_\epsilon}_{\iota \tarrow o} \, x_\iota \Implies
\sembrack{f_\epsilon}_{\iota \tarrow o} \, 
(\mname{S}_{\iota \tarrow \iota} \, x_\iota)))
\Implies 
\ForallApp x_\iota \mdot \sembrack{f_\epsilon}_{\iota \tarrow o} \, x_\iota)
\end{align*}
where $\mname{is-peano}_{\epsilon \tarrow o} \, f_\epsilon$ holds iff
$f_\epsilon$ represents the syntactic tree of a formula of first-order
Peano arithmetic.  Hence it is possible to directly formalize
first-order Peano arithmetic in {\churchqe}.  The \emph{induction
  schema for Presburger arithmetic} can be formalized similarly using
an appropriate predicate $\mname{is-presburger}_{\epsilon \tarrow o}$.

\subsection{Meaning Formulas} \label{secsub:meaning-formulas}

Many symbolic algorithms work by manipulating mathematical expressions
in a mathematically meaningful way.  A \emph{meaning formula} for such
an algorithm is a statement that captures the mathematical
relationship between the input and output expressions of the
algorithm.  For example, consider a symbolic differentiation algorithm
that takes as input an expression (say $x^2$), repeatedly applies
syntactic differentiation rules to the expression, and then returns as
output the final expression ($2x$) that is produced.  The intended
meaning formula of this algorithm states that the function
($\LambdaApp x : \mathbb{R} \mdot 2x$) represented by the output
expression is the derivative of the function ($\LambdaApp x :
\mathbb{R} \mdot x^2$) represented by the input expression.

Meaning formulas are difficult to express in a traditional logic like
first-order logic or simple type theory since there is no way to
directly refer to the syntactic structure of the expressions in the
logic~\cite{Farmer13}.  However, meaning formulas can be easily
expressed in {\churchqe}.

\bsp Consider the following example.  Assume that the domain of the
type~$\iota$ is the real numbers and $\sC$ includes the usual
constants of real number arithmetic plus (1)
$\mname{is-poly}_{\epsilon \tarrow o}$ such that
$\mname{is-poly}_{\epsilon \tarrow o} \, \textbf{A}_\epsilon = \TRUE$
iff $\textbf{A}_\epsilon$ represents a syntax tree of an expression of
type $\iota$ that is a polynomial, (2) $\mname{deriv}_{(\iota \tarrow
  \iota) \tarrow (\iota \tarrow \iota)}$ such that
$\mname{deriv}_{(\iota \tarrow \iota) \tarrow (\iota \tarrow \iota)}
\, \textbf{F}_{\iota \tarrow \iota}$ is the derivative of the function
$\textbf{F}_{\iota \tarrow \iota}$, and (3)
$\mname{poly-diff}_{\epsilon \tarrow \epsilon \tarrow \epsilon}$ such
that, if $\mname{is-poly}_{\epsilon \tarrow o} \, \textbf{A}_\epsilon$
holds, then $\mname{poly-diff}_{\epsilon \tarrow \epsilon \tarrow
  \epsilon} \, \textbf{A}_\epsilon \, \synbrack{\textbf{x}_\iota}$ is
the result of applying the usual differentiation rules for polynomials
to $\textbf{A}_\epsilon$ with respect to $\textbf{x}_\iota$.  Then the
meaning formula for $\mname{poly-diff}_{\epsilon \tarrow \epsilon
  \tarrow \epsilon}$ is
\begin{align*}
&
\ForallApp u_\epsilon \mdot \ForallApp v_\epsilon \mdot
(\mname{is-var}_{\epsilon \tarrow o} \, u_\epsilon
 \Andd
\mname{is-expr}_{\epsilon \tarrow o}^{\iota} \, u_\epsilon
 \Andd
\mname{is-poly}_{\epsilon \tarrow o} \, v_\epsilon) \Implies {} \\
&
\hspace{4ex}
\mname{deriv}_{(\iota \tarrow \iota) \tarrow (\iota \tarrow \iota)}
(\sembrack{\mname{abs}_{\epsilon \tarrow  \epsilon \tarrow \epsilon} \,
u_\epsilon \, v_\epsilon}_{\iota \tarrow \iota}) = {} \\
&
\hspace{4ex}
\sembrack{\mname{abs}_{\epsilon \tarrow  \epsilon \tarrow \epsilon} \,
u_\epsilon \, 
(\mname{poly-diff}_{\epsilon \tarrow \epsilon \tarrow \epsilon} \, 
v_\epsilon \, u_\epsilon)}_{\iota \tarrow \iota}.\footnotemark
\end{align*}

\footnotetext{We restrict this example to polynomials since polynomial
  functions and their derivatives are always total.  Thus issues of
  undefinedness do not arise in the formulation of the meaning formula
  for $\mname{poly-diff}_{\epsilon \tarrow \epsilon \tarrow
    \epsilon}$.}
  
\noindent
The string of equations
\begin{align} \setcounter{equation}{0}
&
\mname{deriv}_{(\iota \tarrow \iota) \tarrow (\iota \tarrow \iota)}
(\LambdaApp x_\iota \mdot x_{\iota}^{2})\\
&=
\mname{deriv}_{(\iota \tarrow \iota) \tarrow (\iota \tarrow \iota)}
(\sembrack{\synbrack{\LambdaApp x_\iota \mdot x_{\iota}^{2}}}_{\iota \tarrow \iota})\\
&=
\mname{deriv}_{(\iota \tarrow \iota) \tarrow (\iota \tarrow \iota)}
(\sembrack{\mname{abs}_{\epsilon \tarrow  \epsilon \tarrow \epsilon} \,
\synbrack{x_{\iota}} \, \synbrack{x_{\iota}^{2}}}_{\iota \tarrow \iota})\\
&=
\sembrack{\mname{abs}_{\epsilon \tarrow  \epsilon \tarrow \epsilon} \,
\synbrack{x_{\iota}} \, 
(\mname{poly-diff}_{\epsilon \tarrow \epsilon \tarrow \epsilon} \, 
\synbrack{x_{\iota}^{2}} \, \synbrack{x_{\iota}})}_{\iota \tarrow \iota}\\
&=
\sembrack{\mname{abs}_{\epsilon \tarrow  \epsilon \tarrow \epsilon} \,
\synbrack{x_{\iota}} \, \synbrack{2 \ast x_{\iota}}}_{\iota \tarrow \iota}\\
&=
\sembrack{\synbrack{\LambdaApp x_\iota \mdot 2 \ast x_{\iota}}}_{\iota \tarrow \iota}\\
&=
\LambdaApp x_\iota \mdot 2 \ast x_{\iota}
\end{align}
proves (informally) the desired result where the equation given by (3)
and (4) results from instantiating the meaning formula for
$\mname{poly-diff}_{\epsilon \tarrow \epsilon \tarrow \epsilon}$ with
$\synbrack{x_{\iota}}$ and $\synbrack{x_{\iota}^{2}}$.  \esp

%% Transformers, meaning formulas, biform theories
%% Natural number arithmetic
%% Polynomials
%% Symbolic differentiation

\section{A Sketch of a Simple Proof System}

\bsp 
At first glance, it would appear that a proof system for {\churchqe}
could be straightforwardly developed by extending Andrews' proof
system for {\qzero}~\cite[p.~213]{Andrews02}.  We can define
$\mname{is-var}_{\epsilon \tarrow o}$ (and $\mname{is-con}_{\epsilon
  \tarrow o}$ in a similar way) by the axiom schemas
$\mname{is-var}_{\epsilon \tarrow o} \, \synbrack{\textbf{x}_\alpha}$
and $\Neg\mname{is-var}_{\epsilon \tarrow o} \, \textbf{A}_\epsilon$
where $\textbf{A}_\epsilon$ is any construction that is not a quoted
variable.  We can recursively define $\mname{is-expr}_{\epsilon
  \tarrow o}^{\alpha}$ using a set of axioms that say how expressions
are constructed.  We can specify that the type $\epsilon$ of
constructions is an inductive type using a set of axioms that say
(1)~the constructions are distinct from each other and (2)~induction
holds for constructions.  We can specify quotation using the Law of
Quotation $\synbrack{\textbf{A}_\alpha} = \sE(\textbf{A}_\alpha)$
(Theorem~\ref{thm:quotation}).  And we can specify evaluation using
the Law of Disquotation
$\sembrack{\synbrack{\textbf{A}_\alpha}}_\alpha = \textbf{A}_\alpha$
(Theorem~\ref{thm:disquotation}).
\esp

Andrews' proof system with these added axioms would enable simple
theorems involving quotation and evaluation to be proved, but the
proof system would not be able to substitute expressions for free
variables occurring in the argument of an evaluation.  Hence schemas
and meaning formulas could be expressed in {\churchqe}, but they would
be useless because they could not be instantiated. Clearly, a useful
proof system for {\churchqe} requires some form of substitution that
is applicable to evaluations.

Due to the Variable Problem, substitution involving evaluations cannot
be purely syntactic as in a traditional logic.  It must be a
semantics-dependent operation in which side conditions, like whether a
variable is free in an expression, are proved within the proof system.
Since {\churchqe} supports reasoning about syntax, an obvious way
forward is to add to $\sC$ a logical constant $\mname{sub}_{\epsilon
  \tarrow \epsilon \tarrow \epsilon \tarrow \epsilon}$ such that, if
$\textbf{C}_\beta$ is the result of substituting $\textbf{A}_\alpha$
for each free occurrence of $\textbf{x}_\alpha$ in $\textbf{B}_\beta$
without any variable captures, then
\[\mname{sub}_{\epsilon \tarrow \epsilon \tarrow \epsilon \tarrow
  \epsilon} \, \synbrack{\textbf{A}_\alpha} \,
\synbrack{\textbf{x}_\alpha} \, \synbrack{\textbf{B}_\beta} =
\synbrack{\textbf{C}_\beta}.\] $\mname{sub}_{\epsilon \tarrow \epsilon
  \tarrow \epsilon \tarrow \epsilon}$ thus plays the role of an
explicit substitution operator~\cite{AbadiEtAl91}.

This approach, however, does not work in {\churchqe} since
$\textbf{B}_\beta$ may contain evaluations, but quotations in
{\churchqe} may not contain evaluations.  Although the approach does
work in {\qzerouqe}~\cite{FarmerArxiv14} in which quotations in
{\churchqe} may contain evaluations, it is extremely complicated due
to the Evaluation Problem.

A more promising approach is to add some axiom schemas to the five
beta-reduction axiom schemas used by Andrews' in his proof system for
{\qzero}~\cite[p.~213]{Andrews02} that specify beta-reduction of an
application of the form $(\LambdaApp \textbf{x}_\alpha \mdot
\sembrack{\textbf{B}_\epsilon}_\beta) \, \textbf{A}_\alpha$.  But how
do we overcome the Double Substitution Problem?  There seems to be no
easy way of emulating a double substitution with beta-reduction, so
the best approach appears to be to consider only cases that do not
require a second substitution, as formalized by the following axiom
schema:
\begin{align*}
&
(\mname{is-expr}_{\epsilon \tarrow o}^{\beta} \, 
((\LambdaApp \textbf{x}_\alpha \mdot \textbf{B}_\epsilon) \, \textbf{A}_\alpha) \Andd 
\Neg(\mname{is-free-in}_{\epsilon \tarrow \epsilon \tarrow o} \,
\synbrack{\textbf{x}_\alpha} \, ((\LambdaApp \textbf{x}_\alpha \mdot
\textbf{B}_\epsilon) \, \textbf{A}_\alpha))) \Implies {} \\
&
\hspace{4ex}
(\LambdaApp \textbf{x}_\alpha \mdot
\sembrack{\textbf{B}_\epsilon}_\beta) \, \textbf{A}_\alpha =
\sembrack{(\LambdaApp \textbf{x}_\alpha \mdot \textbf{B}_\epsilon) \,
  \textbf{A}_\alpha}_\beta.
\end{align*}
Here $\mname{is-free-in}_{\epsilon \tarrow \epsilon \tarrow o}$ would
be a new logical constant in $\sC$, and the second condition would say
that $\textbf{x}_\alpha$ is not free in the expression whose syntax
tree is represented by $(\LambdaApp \textbf{x}_\alpha \mdot
\textbf{B}_\epsilon) \, \textbf{A}_\alpha$.  As a result, there would
be no free occurrences of $\textbf{x}_\alpha$ in the right-hand side
of the conclusion after the evaluation is eliminated.  Details of this
approach will be given in a future paper that presents the proof
system for {\churchqe} that we have sketched.

\section{Conclusion}

Quotation and evaluation provide a basis for metaprogramming as seen
in Lisp and other programming languages.  We believe that these
mechanisms can also provide a basis for metareasoning in traditional
logics like first-order logic or simple type theory.  However,
incorporating quotation and evaluation into a traditional logic is
much more challenging than incorporating them into a programming
language due to the three problems we described in the Introduction.

In this paper we have introduced {\churchqe}, a logic based on
{\qzero}~\cite{Andrews02}, Andrews' version of Church's type theory,
that includes quotation and evaluation.  We have presented the syntax
and semantics of {\churchqe}, sketched a proof system for it, and
given examples that show the practical benefit of having quotation and
evaluation in a logic.

{\churchqe} is a simpler version of {\qzerouqe}~\cite{FarmerArxiv14},
a richer, but more complicated, version of {\qzero} with
undefinedness, quotation, and evaluation.  In {\qzerouqe}, quotation
may be applied to expressions containing evaluations, expressions may
be undefined and functions may be partial, and substitution is
implemented explicitly as a logical constant.  Allowing quotation to
be applied to all expressions makes {\qzerouqe} much more expressive
than {\churchqe} but also much more difficult to implement since
substitution in the presence of evaluations is highly complex.  We
believe that {\churchqe} would not be hard to implement.  Since it is
a version of Church's type theory, it could be implemented by
extending an implementation of HOL~\cite{GordonMelham93} such as
HOL Light~\cite{Harrison09}.

Our approach for incorporating quotation and evaluation into Church's
type theory --- introducing an inductive type of constructions, a
quotation operator, and a typed evaluation operator --- can be applied
to other logics including many-sorted first-order logic.  We have
shown that developing the needed syntax and semantics is relatively
straightforward, while developing a proof system for the logic is
fraught with difficulties.

\section*{Acknowledgments}

The author thanks the reviewers for their helpful comments and
suggestions.

\bibliography{$HOME/research/lib/imps}%$ 

\begin{thebibliography}{10}

\bibitem{AbadiEtAl91}
M.~Abadi, L.~Cardelli, P.-L. Curien, and J.-J. L\'evy.
\newblock Explicit substitution.
\newblock {\em Journal of Functional Programming}, 1:375--416, 1991.

\bibitem{Andrews02}
P.~B. Andrews.
\newblock {\em An Introduction to Mathematical Logic and Type Theory: To Truth
  through Proof, Second Edition}.
\newblock Kluwer, 2002.

\bibitem{Bawden99}
A.~Bawden.
\newblock Quasiquotation in {Lisp}.
\newblock In O.~Danvy, editor, {\em Proceedings of the 1999 ACM SIGPLAN
  Symposium on Partial Evaluation and Semantics-Based Program Manipulation},
  pages 4--12, 1999.
\newblock Technical report BRICS-NS-99-1, University of Aarhus, 1999.

\bibitem{Chlipala13}
A.~Chlipala.
\newblock {\em Certified Programming with Dependent Types: {A} Pragmatic
  Introduction to the Coq Proof Assistant}.
\newblock MIT Press, 2013.

\bibitem{Church40}
A.~Church.
\newblock A formulation of the simple theory of types.
\newblock {\em Journal of Symbolic Logic}, 5:56--68, 1940.

\bibitem{Costantini02}
S.~Costantini.
\newblock Meta-reasoning: {A} survey.
\newblock In A.~C. Kakas and F.~Sadri, editors, {\em Computational Logic: Logic
  Programming and Beyond, Essays in Honour of Robert A. Kowalski, Part {II}},
  volume 2408 of {\em Lecture Notes in Computer Science}, pages 253--288, 2002.

\bibitem{FarmerArxiv13}
W.~M. Farmer.
\newblock Chiron: {A} set theory with types, undefinedness, quotation, and
  evaluation.
\newblock {\em Computing Research Repository}, abs/1305.6206 (154 pp.), 2013.

\bibitem{Farmer13}
W.~M. Farmer.
\newblock The formalization of syntax-based mathematical algorithms using
  quotation and evaluation.
\newblock In J.~Carette, D.~Aspinall, C.~Lange, P.~Sojka, and W.~Windsteiger,
  editors, {\em Intelligent Computer Mathematics}, volume 7961 of {\em Lecture
  Notes in Computer Science}, pages 35--50. Springer, 2013.

\bibitem{FarmerArxiv14}
W.~M. Farmer.
\newblock Simple type theory with undefinedness, quotation, and evaluation.
\newblock {\em Computing Research Repository}, abs/1406.6706 (87 pp.), 2014.

\bibitem{GonthierEtAl15}
G.~Gonthier, A.~Mahboubi, and E.~Tassi.
\newblock {A Small Scale Reflection Extension for the Coq system}.
\newblock Research Report RR-6455, {Inria Saclay Ile de France}, 2015.

\bibitem{GordonMelham93}
M.~J.~C. Gordon and T.~F. Melham.
\newblock {\em Introduction to {HOL}: A Theorem Proving Environment for Higher
  Order Logic}.
\newblock Cambridge University Press, 1993.

\bibitem{GrundyEtAl06}
J.~Grundy, T.~Melham, and J.~O'Leary.
\newblock A reflective functional language for hardware design and theorem
  proving.
\newblock {\em Journal of Functional Programming}, 16, 2006.

\bibitem{Harrison95}
J.~Harrison.
\newblock Metatheory and reflection in theorem proving: A survey and critique.
\newblock Technical Report CRC-053, SRI Cambridge, 1995.
\newblock Available at
  \url{http://www.cl.cam.ac.uk/~jrh13/papers/reflect.ps.gz}.

\bibitem{Harrison09}
J.~Harrison.
\newblock {HOL Light}: {A}n overview.
\newblock In S.~Berghofer, T.~Nipkow, C.~Urban, and M.~Wenzel, editors, {\em
  Theorem Proving in Higher Order Logics}, volume 5674 of {\em Lecture Notes in
  Computer Science}, pages 60--66. Springer, 2009.

\bibitem{Henkin50}
L.~Henkin.
\newblock Completeness in the theory of types.
\newblock {\em Journal of Symbolic Logic}, 15:81--91, 1950.

\bibitem{Norell07}
U.~Norell.
\newblock {\em Towards a Practical Programming Language based on Dependent Type
  Theory}.
\newblock PhD thesis, Chalmers University of Technology, 2007.

\bibitem{Norell09}
U.~Norell.
\newblock Dependently typed programming in {Agda}.
\newblock In A.~Kennedy and A.~Ahmed, editors, {\em TLDI}, pages 1--2. ACM,
  2009.

\bibitem{Elixir15}
Plataformatec.
\newblock Elixir.
\newblock \url{http://elixir-lang.org/}, 2015.

\bibitem{Quine03}
W.~V.~O. Quine.
\newblock {\em Mathematical Logic: Revised Edition}.
\newblock Harvard University Press, 2003.

\bibitem{MetaOCaml11}
{Rice University Programming Languages Team}.
\newblock Metaocaml: A compiled, type-safe, multi-stage programming language.
\newblock \url{http://www.metaocaml.org/}, 2011.

\bibitem{SheardJones02}
T.~Sheard and S.~P. Jones.
\newblock Template meta-programming for {Haskell}.
\newblock {\em ACM SIGPLAN Notices}, 37:60--75, 2002.

\bibitem{Stump09}
A.~Stump.
\newblock Directly reflective meta-programming.
\newblock {\em Higher-Order and Symbolic Computation}, 22:115--144, 2009.

\bibitem{TahaSheard00}
W.~Taha and T.~Sheard.
\newblock Meta{ML} and multi-stage programming with explicit annotations.
\newblock {\em Theoretical Computer Science}, 248:211--242, 2000.

\bibitem{Tarski35a}
A.~Tarski.
\newblock The concept of truth in formalized languages.
\newblock In J.~Corcoran, editor, {\em Logic, Semantics, Meta-Mathematics},
  pages 152--278. Hackett, second edition, 1983.

\bibitem{FSharp15}
{The F\# Software Foundation}.
\newblock F\#.
\newblock \url{http://fsharp.org/}, 2015.

\bibitem{VanDerWaltSwierstra12}
P.~Van~Der Walt and W.~Swierstra.
\newblock Engineering proof by reflection in {Agda}.
\newblock In R.~Hinze, editor, {\em Implementation and Application of
  Functional Languages}, volume 8241 of {\em Lecture Notes in Computer
  Science}, pages 157--173. Springer, 2012.

\end{thebibliography}
\bibliographystyle{plain}

\end{document}